\documentclass[11pt]{article}
\usepackage[margin=1in]{geometry}
\usepackage{amsmath,amssymb,amsthm,color}
\title{Minimax Option Pricing Meets Black-Scholes in the Limit}

\author{
Jacob Abernethy\\
Computer Science Division\\
University of California, Berkeley\\
\texttt{\small jake@cs.berkeley.edu}
\and
Rafael M.\ Frongillo\\
Computer Science Division\\
University of California, Berkeley\\
\texttt{\small raf@cs.berkeley.edu}
\and
Andre Wibisono\\
Computer Science Division\\
University of California, Berkeley\\
\texttt{\small wibisono@cs.berkeley.edu}
}
\date{}

\newcommand{\Comments}{1}
\newcommand{\note}[2]{\ifnum\Comments=1\textcolor{#1}{#2}\fi}

\def\F{\mathcal{F}}
\def\A{\mathcal{A}}
\def\C{\mathcal{C}}

\def\E{\mathbb{E}}
\def\Eop{\mathop{\mathbb{E}}}

\def\reals{\mathbb{R}}

\def\R{\mathbb{R}}

\def\N{\mathbb{N}}
\def\Var{\text{Var}}
\def\C{\mathcal{C}}
\def\X{\mathcal{X}}
\def\P{\mathbb{P}}
\def\S{\mathcal{S}}

\def\zb{\boldsymbol{\zeta}}

\newcommand{\convto}[1]{\stackrel{#1}{\longrightarrow}}
\newcommand{\ZC}{\text{ZC$_n$}}   % zeta constraint
\newcommand{\CVC}{\text{CVC}} % continuous variance constraint
\newcommand{\DVC}{\text{DVC}} % discrete variance constraint
\newcommand{\mg}{\textnormal{mg}}

\newtheorem{lemma}{Lemma}

\newtheorem{theorem}{Theorem}
\newtheorem{definition}{Definition}

\theoremstyle{definition}

\begin{document}

\maketitle

\begin{abstract}
  Option contracts are a type of financial derivative that allow investors to hedge risk and speculate on the variation of an asset's future market price. In short, an option has a particular payout that is based on the market price for an asset on a given date in the future. In 1973, Black and Scholes proposed a valuation model for options that essentially estimates the tail risk of the asset price under the assumption that the price will fluctuate according to geometric Brownian motion. More recently, DeMarzo et al., among others, have proposed more robust valuation schemes, where we can even assume an adversary chooses the price fluctuations. This framework can be considered as a sequential two-player zero-sum game between the investor and Nature. We analyze the value of this game in the limit, where the investor can trade at smaller and smaller time intervals. Under weak assumptions on the actions of Nature (an adversary), we show that the minimax option price asymptotically approaches exactly the Black-Scholes valuation. The key piece of our analysis is showing that Nature's minimax optimal dual strategy converges to geometric Brownian motion in the limit.
\end{abstract}

\newpage

\section{Introduction}

In finance, an \emph{option} is a financial contract between two parties that guarantees the purchase or sale of a given asset, such as a stock or bond, at a specified price in the future. Buying an option means buying the right to engage in a particular transaction, yet the the buyer has no obligation to do so. Options are a useful tool for controlling risk in financial portfolios, as they can be used to hedge against the possibility of a large unexpected price fluctuation.

Let us focus presently on a \emph{European call option}, parameterized by an asset $A$, a \emph{strike price} $K$ and a future \emph{expiration date} $T$. The buyer of an $(A,K,T)$ call option is allowed to purchase 1 share of asset $A$ for a fixed price of $K$ on date $T$ if she so chooses. Of course, if the market price of $A$ on date $T$ is $P$, then the owner of the call option will only exercise the transaction if $P > K$. Thus, in practice, the value of the option on date $T$ is $\max(0, P - K)$.

What has remained a popular topic in finance is the problem of pricing options in terms of known properties of the underlying asset and distributional properties of its price fluctuations. The most well-known approach for pricing options is the \emph{Black-Scholes} model, introduced in 1973 by Fischer Black and Myron Scholes in their seminal paper ``The pricing of options and corporate liabilities'' \cite{black1973pricing}. As the future market price of an asset is an uncertain quantity, the Black-Scholes pricing model includes a key assumption, that an asset's price fluctuates with constant drift and volatility, which leads to a \emph{geometric Brownian motion} (GBM) model for the price path. Thus, given a European call option $(A,K,T)$, we can estimate the ``fair value'' of the option to be its expected payoff under the assumption that the price path $P(t)$ behaves according to GBM. That is, the Black-Scholes model would set the option price to be $\E_{P \sim \text{GBM}} [\max(0, P(T) - K)]$. 

The Black-Scholes model has undergone a reasonable amount of criticism, much of which is due to the GBM characterization of the underlying asset price, a model which heavily discounts tail risk. Extreme price changes, often due to systemic events, are estimated to be highly improbable under Black-Scholes but occur quite frequently in practice -- this valuation model does not hold for alternative stochastic models.

There has been recent work that takes an entirely different approach to pricing options, namely where no stochastic assumptions are made about the fluctuation of the underlying asset. This robust option pricing framework, put forward by DeMarzo et al.\ \cite{demarzo2006online} (with a similar model presented by Shafer and Vovk \cite{shafer2001probability}) imagines a multi-stage game between an investor and Nature. The investor chooses an amount of money to invest in the underlying asset, and Nature chooses how the asset's price should change from round to round. The investor would like to exhibit a trading strategy (algorithm) with the ultimate goal of earning almost as much as the payout of the option \emph{against a worst-case price path}. Let us imagine, for the moment, that we can construct a strategy which receives a payout that is never more than $C$ dollars worse than the option payout. DeMarzo et al.\ make the key observation that, since this guarantee holds for any price path (within constraints), then the price of the option should cost no more than $C$ at the start of the game. The reason for this is simple: if the price of the option were strictly larger than $C$, then the investor has an arbitrage opportunity via short selling the option and going long on his robust trading strategy.

In the present paper, we look at this sequential game between investor and Nature and analyze the equilibrium strategies of each player. But we go a step further and consider what is the limit behavior of the game when the investor trades at greater and greater frequency. Intuitively, this is what would happen when a firm switches from a trading strategy that trades once per day, to one that trades once per hour, to once per minute, etc. The question we aim to answer is: what is the value of this game, i.e.\ the minimax option price, as the trading frequency approaches infinity? The answer is somewhat surprising: the option value, in the limit, is \emph{identical to the price under the Black-Scholes model.} We show this by proving that the worst case price path chosen by Nature, under reasonable constraints, will converge to geometric Brownian motion.

Recall that the Black-Scholes option pricing model is effectively about determining a ``fair price'' for the derivative, under reasonably strong distributional constraints. The robust option pricing framework of DeMarzo et al., on the other hand, only aims to exhibit an upper bound on the option value when Nature sets the market price of the asset under some constraints on the price fluctuations. But we show that the worst-case price is identically the fair price under the GBM assumption. In a certain sense this provides reasonably strong validation for Black-Scholes.

We make two observations about the results.
\begin{enumerate}
  \item While our running an example is a European call option, our main theorem holds for any financial derivative whose payoff is a convex function of the underlying asset price on the expiration date. This is a broad class of derivatives, including both the European put and call options.
  \item We do not provide an explicit optimal strategy for the investor. We begin by considering a game where first the investor commits to an algorithm and then Nature responds with a randomized price path, but the majority of our analysis concerns the dual of this game in which Nature must act first, and then the investor can act with full knowledge of Nature's randomized price path. We leave as an open question whether such a strategy can be efficiently constructed.
  \item Our analysis requires that the investor can trade arbitrarily large amounts of money to compete with the option payout. In a sense, we are assuming that the investor is a large institution with the captial to go very long and very short on the asset. We believe that the result should approximately hold when the investor has a fixed budget size, but we leave this as an open question. 
\end{enumerate}

%-----------
\subsection*{Previous Work}

The primary motivation for our results draws from DeMarzo et al.\ \cite{demarzo2006online}. This work focused primarily on constructing robust trading strategies that can be used instead of purchasing an option, and the authors were the first to observe that the existence of such a strategy provides a theoretical upper bound on the price one should pay for an option. Their results draw strongly from ideas developed within the theory of \emph{regret minimization} in online learning, and the proposed algorithm has a resemblance to the well-known Multiplicative Weights algorithm, proposed in various forms by, among others, Littlestone and Warmuth \cite{LitWar94} and Freund and Schapire \cite{freund1995desicion}. The connection to regret minimization strategies is unsurprising, since the primary goal of this area of research is to provide guarantees that hold without requiring stochasticity assumptions on the data received from Nature.

The work of DeMarzo et al.\ looks briefly at the minimax formulation of the option pricing game, although for a fixed trading frequency and using stronger assumptions than what is required in our results. Their analysis requires using Sion's minimax theorem \cite{sion1958general} and they conclude that Nature's optimal price path must be a martingale; both of these components are used herein. Finally, the authors provide a plot (Figure 1 in Section 7 in \cite{demarzo2006online}) showing the computed minimax option price versus the Black-Scholes price as a function of the strike price. This plot strongly foreshadows our main result in this paper as the two curves are very close to one another, and our aim is to show that these curves are asymptotically equal.

The other very related work is the book by Shafer and Vovk \cite{shafer2001probability} titled ``Probability and finance: it's only a game!'' The authors considered a number of game-theoretic interpretations of problems in finance. The authors look at pricing options under a game-theoretic framework (introduced originally by Vovk \cite{vovk1995pricing}) in which an investor must make a sequence of trades with an underlying asset as Nature sets the asset's market price. They arrive at very similar conclusions to ours, establishing that the ``fair price'' under their model is identical to the Black-Scholes valuation. However, the framework they put forward is more akin to the original derivation of the Black-Scholes pricing model, and differs from the framework of DeMarzo et al.\ in that they do not aim for worst-case bounds. In particular, their analysis requires the existence of a hypothetical derivative which pays off according to the fluctuation of the underlying asset.

There is a significant amount of discussion and analysis of sequential minimax games in the book of Cesa-Bianchi and Lugosi \cite{cesa-bianchi_prediction_2006}. A full duality-based approach to analyzing repeated games can be found in Abernethy et al.\ \cite{abernethy2009stochastic}, who also utilize martingale concentration. There has been more recent and very interesting work on pricing more exotic options by Gofer and Mansour \cite{gofer2011pricing,gofer2011regret}.

%-----------
\section{The Minimax Option Game}

\paragraph{Notation.} Let $\R_0 = [0,\infty)$, $\R_+ = (0,\infty)$, and $\N = \{1,2,3,\dots\}$. Throughout, $m,n$ index $\N$ and $t$ indexes $\R$, so for instance, $0 \leq m \leq n$ means $m \in \{0,1,\dots,n\}$ and $0 \leq t \leq 1$ means $t \in [0,1]$.

\subsection{Problem setting and definitions.}
We shall be considering the value of options (or other derivative contracts) providing a certain payoff that depends on (a) the expiration date $T$, (b) a specified ``strike price'' $K$ and (c) the price $X(T)$ of some underlying asset at time $T$. We can assume that the current price $X(0)$ is known in advance, hence we let $X(0) = 1$ always. Since we consider trading in continuous time, we shall assume without loss of generality that the option expires at $T = 1$. We shall consider, for each $n \in \N$, a sequential game where an investor trades $n$ times throughout the time interval $[0,1]$, i.e.\ we trade at every $1/n$-time interval.

Let us now discuss the payoff of the option, which will be denoted by a function $g \colon \R_0 \to \R_0$ whose input is the asset price at the time of expiration. For example, in the case of the European call option the payoff function is $g(x) = \max(0,x-K)$ where $K \geq 0$ is the strike price. We do not restrict our attention to derivatives of this form, and instead we require only that $g$ is convex and $L$-Lipschitz; namely, $|g(y)-g(x)| \leq L|y-x|$ for all $x,y \in \R_0$. However, for simplicity, we use the term ``option'' throughout the paper.

In the option pricing game we consider, we imagine that Nature chooses a randomized price path for the underlying asset, with the goal of maximizing the expected difference between the payoff of the option and the investor's earnings. We generally use the symbol $X$ to denote this price path. Let $\X$ denote the set of continuous stochastic processes $X \colon [0,1] \to \R_0$ with $X(0) = 1$, representing the asset prices. Within our analysis, we also consider a discrete variant of the continuous game, and imagine Nature choosing price paths defined by a finite sequence of random variables. For each $n \in \N$, let $\S^n$ denote the set of sequences of random variables $S_n = (S_{n,0}, S_{n,1}, \dots, S_{n,n})$ with $S_{n,0} = 1$ and $S_{n,m} \in \R_+$ for $0 \leq m \leq n$.

Given $X \in \X$, we will write $S_n(X) = (X(0), X(1/n), \dots, X(n/n)) \in \S^n$ to denote the discrete points of $X$.  Similarly, given $S_n\in\S^n$, we write $X(S_n)$ to denote the stochastic process $X \in \X$ obtained by linearly interpolating the values of $S_n$ in the log space, i.e.\ for any $t \in [0,1]$ of the form $t = (m + \alpha)/n$ with $m \in \{0,1,\dots,n-1\}$ and $0 \leq \alpha \leq 1$, we define
\begin{equation*}
\log X(t) = (1-\alpha) \: \log S_{n,m} + \alpha \: \log S_{n,m+1}.
\end{equation*}

Let $B \colon [0,1] \to \R$ denote the standard Brownian motion (also known as a Weiner Process) with $B(0) = 0$, and let $G \colon [0,1] \to \R_+$ denote the geometric Brownian motion (GBM) with drift $0$ and volatility $\sqrt{c}$,
\begin{equation} \label{eq:gbmdef}
G(t) = \exp\left(\sqrt{c} \: B(t) - \frac{ct}{2} \right).
\end{equation}
Observe that $G$ is a martingale since it has drift $0$. We consider $B$ and $G$ as (random) elements of $\C[0,1]$, the space of continuous functions from $[0,1]$ to $\R$
equipped with the supremum norm.

When the price path behaves according to GBM, we can define the Black-Scholes formulation for pricing an option defined by payoff function $g$:
\[
  (\text{Black-Scholes Price}) \quad \beta \quad := \quad \E[g(G(1))],
\]
where $G$ has the distribution defined as in \eqref{eq:gbmdef}, and the constant $c$ is assumed to be fixed. We observe one property of standard Brownian motion, which is that the value $B(t)$ conditioned on $B(t_0)$ for some $t_0 < t$ is normally distributed with mean $B(t_0)$ and variance $t - t_0$. Similarly, it is easy to see that
\begin{equation} \label{eq:gbmvar}
  \E\left[ (G(t) - G(t_0))^2 \, | \, \F_{t_0} \right] = G(t_0)^2 (\exp(c(t-t_0)) - 1)
\end{equation}
where $\F_s = \sigma(G(r) \colon 0 \leq r \leq s)$ is the filtration generated by $G$ up to time $s$.

We turn our attention to setting constraints for Nature's choice of price path $X \in \X$ -- it is very difficult to obtain any reasonable results unless the adversary is constrained in some fashion. In light of \eqref{eq:gbmvar}, its natural to require that the \emph{logarithm} of the price path not fluctuate too greatly.
\begin{definition}[\CVC]
We say that $X\in\X$ satisfies the {\em continuous variance constraint} (\CVC) if
\begin{equation} \label{eq:varbound}
\E[(X(t)-X(s))^2 \mid \F_s] \leq \left(\exp(c (t-s))-1\right) \: X(s)^2 \quad \text{a.s.\ for all } 0 \leq s \leq t \leq 1,
\end{equation} 
where $\F_s$ is the filtration generated by $X$ up to time $s$, and $c$ is a variance parameter.
\end{definition}
We will implicitly assume throughout that the variance parameter $c > 0$ is fixed. Observe that the geometric Brownian motion satisfies \CVC\ with equality. Since we shall also consider discrete price paths $S_n = (S_{n,0}, S_{n,1}, \dots, S_{n,n})$, we construct a similar constraint for this case.
\begin{definition}[\DVC]
For each $n \in \N$, we say $S_n \in \S^n$ satisfies the {\em discrete variance constraint} (\DVC) if
\begin{equation}\label{eq:dvcbound}
\E[(S_{n,m+1} - S_{n,m})^2  \mid \F_{n,m}] \leq (\exp(c/n) - 1) \: S_{n,m}^2 \quad \text{a.s.\ for all } m = 0,1,\dots,n-1,
\end{equation}
where now $\F_{n,m} = \sigma(S_{n,0}, \dots, S_{n,m})$.
\end{definition}
It is worth noting that, while the \CVC\ and \DVC\ apply to different spaces, the latter is essentially weaker than the former. Rather than require that the variance be controlled for every $s,t$ time interval, \DVC\ only specifies that the fluctuations between successive $1/n$ length time points are bounded. This fact is used in our proof.

We will need one more constraint on Nature's choice of price path. To describe this final restriction, let $\zb = (\zeta_1,\zeta_2,\dots)$ be a sequence of positive real numbers with the property that $\zeta_n \to 0$ and
\begin{equation*}
\liminf_{n \to \infty} \frac{n \zeta_n^2}{\log n} > 16c.
\end{equation*}
For example, we can take $\zeta_n = n^{-\frac{1}{2}+\delta}$ for any $0 < \delta < 1/2$. We will assume throughout the paper that $\zb$ is chosen and fixed, and satisfies the above properties.
\begin{definition}[\ZC]
For each $n \in \N$, we say that $S_n\in \S^n$ satisfies the {\em $\zeta_n$-constraint} (\ZC) if
\begin{equation} \label{eq:zc}
\left| \frac{S_{n,m+1}}{S_{n,m}} - 1 \right| \leq \zeta_n \quad \text{a.s.\ for all } m = 0,1,\dots,n-1. 
\end{equation}
Similarly, for each $n \in \N$ we say that $X\in \X$ satisfies the {\em $\zeta_n$-constraint} if $S_n(X)$ does.
\end{definition}

The \ZC\ constraint may appear strong, since we have a hard bound on what values the price ratios can take, but upon closer inspection one sees that this is a very weak constraint. It is required for our results because we require in various places that the prices lie in a compact set. However, the bounds we get are independent of the sequence $\zb$. And furthermore while the $\zeta_n$ must approach $0$, it can do so at an arbitrarily slow rate. Notice that \eqref{eq:dvcbound} is a similar constraint to that of \eqref{eq:zc}, where the former is ``soft'' and the latter ``hard'', but the constraint in \eqref{eq:dvcbound} shrinks at a rate of $\Theta(1/n)$ whereas the rate in \eqref{eq:zc} shrinks at a much slower rate. We will show that the addition of \ZC\ becomes negligible in the limit. While GBM violates \ZC, we show in \eqref{eq:probconv} that it does so with vanishing probability.

\begin{definition}
Define the following sets:
\begin{itemize}
  \item $\X_C = \{X\in \X: X \text{ satisfies \CVC}\}$.
  \item $\X^n_{C,\zb} = \{X\in \X_C: S_n(X) \text{ satisfies \ZC}\}$.
  \item $\S^n_D = \{S_n\in\S^n: S_n \text{ satisfies \DVC}\}$.
  \item $\S^n_{D,\zb} = \{S_n\in\S^n_D: S_n \text{ satisfies \ZC}\}$.
  \item $\S^n_{D,\zb,\mg} = \{S_n\in\S^n_{D,\zb}: S_n \text{ is a martingale with respect to the natural filtration } (\F_{m,n})\}$.
  \item $\S^n_{D=,\zb,\mg} = \{S_n\in\S^n_{D,\zb,\mg}: S_n \text{ satisfies \DVC\ with equality a.s. for all $m$}\}$.
\end{itemize}
\end{definition}
Note that we have the following relations:
\begin{eqnarray*}
&\X^n_{C,\zb} \subseteq \X_C \subseteq \X& \\
&\S^n_{D=,\zb,\mg} \subseteq \S^n_{D,\zb,\mg} \subseteq \S^n_{D,\zb} \subseteq \S^n_D \subseteq \S^n&
\end{eqnarray*}
Furthermore, observe that if $X \in \X_C$ then $S_n(X) \in \S^n_D$ for each $n \in \N$.

\subsection{The Option Pricing Game.}
We wish to analyze the zero-sum game between the investor, who buys and sells shares in an underlying asset, and Nature (the adversary), who chooses the path of the asset's market price over the time period. We shall assume that the game is parameterized by a value $n$ which determines the frequency of the investor's trades. More specifically, we'll allow the investor to adjust his investment after each $1/n$ interval of time.

\paragraph{Continuous Pricing Game:} 
\begin{itemize}
\item Investor's strategy $A \in \A_n$ is a tuple of functions $A_1, \ldots, A_n$ each having the form $A_m: X|_{[0,(m-1)/n]}\mapsto \Delta_m$ where $\Delta_m \in \reals$ and $X|_{[0,t]}$ is the stochastic process $X \in \X$ restricted to the range $[0,t]$. In other words, the investor will choose an amount of money $\Delta_m$ to invest in the underlying asset after having observed the price path up to time $(m-1)/n$. That is, $A$ can be interpreted as a sequence of random variables $(\Delta_1, \Delta_2,\dots,\Delta_n)$ where $\Delta_m$ is measurable with respect to $\F_{\frac{m-1}{n}}$, the filtration of the price path $X$ up to time $\frac{m-1}{n}$.
\item We imagine Nature's strategy as selecting a price path $X \in \X$ satisfying both \CVC\ and \ZC. More simply, we let Nature's strategy set be $\X^n_{C,\zb}$.
\item Assume now that the investor has committed to a strategy $A$ and Nature has committed to a price path $X$. At round $m$, the investor had invested $\Delta_m$ units of currency in the underlying asset previously, and the price fluctuated from $X((m-1)/n)$ to $X(m/n)$. Hence, in this round the investor has earned exactly $\left( \frac{X(m/n)}{X((m-1)/n)} - 1 \right)\Delta_m$.
\item Recall that the investor's goal is to construct a strategy that can compete with the payout of the option. Imagine, for example, that the investor had a strategy to buy and sell the underlying asset as he watches the price fluctuate, and that even in the worst case this strategy had a payout that was only $D$ dollars worse than the payout of the option. Then the investor would never pay more than a price of $D$ to purchase the option in question, for he could simple execute his successful strategy and never lose more than $D$. This argument is one of the key ideas in DeMarzo et al.\ \cite{demarzo2006online}.
\item With the previous observation in mind, we should design our option pricing game so as to determine the largest deviation between the option payout and the earnings of the investor when the investor is playing the optimal strategy and Nature is selecting a worst case price path. For a particular trading strategy $A \in \A_n$ and a price path $X \in \X$, we define the the loss
\[
  L_n(A,X) = \E_{X} \left[ g(X(1)) - 
    \sum_{m=1}^{n} \left( \frac{X(m/n)}{X((m-1)/n)} - 1 \right)\Delta_m \right].
\]
It will be convenient to consider what happens when the sum $\sum_{m=1}^{n} \left( \frac{X(m/n)}{X((m-1)/n)} - 1 \right) \Delta_m$ vanishes. For this case the investors strategy is irrelevant, hence we define
\[
  L_n(X) = L(X) := \E_{X} [g(X(1))].
\]
\end{itemize}
With the previous discussion in mind, we can now precisely define the quantity of interest. The central focus of the present work is to study the asymptotic value of the discussed game, which is exactly:
\begin{equation}
  \lim_{n \to \infty} \inf_{A\in\A_n} \sup_{X \in \X^n_{C,\zb}} L_n(A,X).
\end{equation}

\paragraph{Discrete Pricing Game} In order to analyze this game we will also consider a discrete version, in which Nature selects some random sequence $S_n \in \S^n_{D,\zb}$ instead of a full price path in $\X$. The set of algorithms $\A_n$ need not be redefined for this discrete game, as we shall assume that on some input $S$ an algorithm $A$ trades according to continuous price path $X(S_n)$ obtained by interpolation defined above. For a given $S_n \in \S^n_{D,\zb}$, we abuse notation somewhat by defining
\[
  L_n(A,S_n) := L_n(A, X(S_n)) \quad \text{ and } \quad L_n(S_n) := L_n(X(S_n)) = \E[g(S_{n,n})].
\]
This is perfectly natural since, in the continuous version, the loss function can be computed by only looking at the points $X(0), X(1/n), \ldots, X(n/n)$.

%%%%%%%%%%%
\section{The Main Result}
\label{sec:main-result}

We now state our main result and given the skeleton of the proof, with the more challenging lemmas saved for the appendix. The proof has a number of interesting ingredients, from an application of Sion's minimax theorem (Lemma~\ref{lem:minimax}), a version of the ``maximum principle'' for maximization of convex functions with random inputs (Lemma~\ref{lem:maxvar}), and a lower bound that requires analyzing Gaussian tails (Lemma~\ref{lem:lower-bound}). But the heaviest lifting is done in Theorem~\ref{Thm:Convergence}, which utilizes a key application of the Lindeberg--Feller Theorem for martingale convergence.

\begin{theorem}
  We have
  \[
    \lim_{n \to \infty} \inf_{A\in\A_n} \sup_{X \in \X^n_{C,\zb}} L_n(A,X) = \beta,
  \]
  where $\beta = \E[g(G(1))]$ is the Black-Scholes price.
\end{theorem}
\begin{proof}
Observe that for sufficiently large $n \in \N$,
\begin{align*}
\inf_{A\in\A_n} \sup_{X \in \X^n_{C,\zb}} L_n(A,X)
&= \sup_{X \in \X^n_{C,\zb}} \inf_{A\in\A_n} L_n(A,X) & \text{ by Lemma~\ref{lem:minimax}} \\
&\leq \sup_{S_n \in \S^n_{D,\zb}} \inf_{A\in\A_n} L_n(A,S_n) & \text{ since $S_n(\X^n_{C,\zb}) \subseteq \S^n_{D,\zb}$} \\
&= \sup_{S_n \in \S^n_{D,\zb,\mg}} L_n(S_n) & \text{ by Lemma~\ref{lem:mtg}} \\
&= L_n(S_n^\ast) & \text{ by Lemma~\ref{lem:maxvar}}
\end{align*}
where $S_n^\ast$ is an element of $\S^n_{D=,\zb,\mg}$ that achieves the $\sup_{\S^n_{D,\zb,\mg}}$ (the existence of which is proven in Lemma~\ref{lem:maxvar}). By letting $n \to \infty$ and using Lemma~\ref{lem:upper-bound} we obtain
\begin{equation*}
\limsup_{n \to \infty} \inf_{A\in\A_n} \sup_{X \in \X^n_{C,\zb}} L_n(A,X)
\leq \lim_{n \to \infty} L_n(S_n^\ast)
= \beta.
\end{equation*}
On the other hand, by Lemma~\ref{lem:lower-bound} we also know that
\begin{equation*}
\beta \leq \liminf_{n \to \infty} \inf_{A\in\A_n} \sup_{X \in \X^n_{C,\zb}} L_n(A,X).
\end{equation*}
Hence we conclude that
\begin{equation*}
\lim_{n \to \infty} \inf_{A\in\A_n} \sup_{X \in \X^n_{C,\zb}} L_n(A,X) = \beta.
\end{equation*}
\end{proof}

We now proceed to establish the necessary lemmas for the above proof. The first lemma states that we have minimax duality for our game.

\begin{lemma}\label{lem:minimax}
  For every $n$ we have
  \[
    \inf_{A\in\A_n} \sup_{X \in \X^n_{C,\zb}} L_n(A,X) = 
      \sup_{X \in \X^n_{C,\zb}} \inf_{A\in\A_n} L_n(A,X).
  \]
\end{lemma}
This lemma is proved in the appendix, but it follows essentially from the minimax theorem of Sion \cite{sion1958general}.

\newcommand{\eqnspacer}{\phantom{\sup_{T_{n,n}}\hspace{-15pt}}}

The second lemma states that the optimal strategy for nature is a martingale, and furthermore, that the strategy of the investor does not matter.

\begin{lemma}\label{lem:mtg}
  \[
  \sup_{S_n \in \S^n_{D,\zb}} \inf_{A\in\A_n} L_n(A,S_n) =
  \sup_{S_n \in \S^n_{D,\zb,\mg}} L_n(S_n),
  \]
\end{lemma}
\begin{proof}
  Defining $T_{n,m} = \frac{S_{m,n}}{S_{n,m-1}} - 1$, we can ``unwind'' the game round-by-round:
  \begin{multline} \label{eq:unwound}
    \sup_{S_n \in \S^n_{D,\zb}} \inf_{A\in\A_n} L_n(A,S_n)
    \quad = \quad
    \sup_{T_{n,1}} \inf_{\Delta_1} \Eop_{S_{n,1}} \left[ -T_{n,1} \Delta_1 \;+\;
      \sup_{T_{n,2}} \inf_{\Delta_2} \Eop_{S_{n,2}} \left[
        - T_{n,2}\Delta_2  \eqnspacer \right.\right.  \\
      + \cdots +
      \left.\left. \sup_{T_{n,n}} \inf_{\Delta_n} \Eop_{S_{n,n}} \left[
        -T_{n,n} \Delta_n + g(S_{n,n}) \eqnspacer
        \right] \cdots \right] \right],
  \end{multline}
  where for all $m$, $\E_{S_{n,m}}[ \,\cdot\, ]$ should be read as $\E_{S_{n,m}}[ \,\cdot\, |S_{n,m-1}]$.  Suppose $S_n$ was such that $\E_{S_{n,m}}[T_{n,m}\,|\, S_{n,m-1}] = a \neq 0$ for some $m$; then we see that 
\[
  \inf_{\Delta_m\in\reals} \Eop_{S_{n,m}}[-\Delta_m T_{n,m} \;|\; S_{n,m-1}] = \inf_{\Delta_m\in\reals} -a\Delta_m = -\infty.
\]
From~\eqref{eq:unwound}, it now follows that $\inf_{A\in\A_n} L_n(A,S_n) = -\infty$ for such an $S_n$.  Thus, to ensure $L_n(A,S_n) > -\infty$, the adversary must set $\E[T_{n,m} | S_{n,m-1}] = 0$ for each $m$, meaning $S_n$ must be a martingale sequence by definition of $T_{n,m}$.  Hence,
\begin{equation*}
  \sup_{S_n \in \S^n_{D,\zb}} \inf_{A\in\A_n} L_n(A,S_n) =
  \sup_{S_n \in \S^n_{D,\zb,\mg}} \inf_{A\in\A_n} L_n(A,S_n).
\end{equation*}
Furthermore, since $S_n$ is a martingale, we see from~\eqref{eq:unwound} that the investor's actions $\Delta_m$ are irrelevant. In particular, we can write
\begin{equation}
  \sup_{S_n \in \S^n_{D,\zb,\mg}} \inf_{A\in\A_n} L_n(A,S_n) = \sup_{S_n \in \S^n_{D,\zb,\mg}} L_n(S_n),\label{eq:player-irrelevant}
\end{equation}
which concludes the proof.
\end{proof}

The third lemma states that the supremum of the objective function is achieved by a stochastic process with maximal variance.

\begin{lemma}\label{lem:maxvar}
For sufficiently large $n$, there exists $S_n^\ast \in \S^n_{D=,\zb,\mg}$ such that
\begin{equation*}
L_n(S_n^\ast) = \sup_{S_n \in \S^n_{D,\zb,\mg}} L_n(S_n).
\end{equation*}
\end{lemma}

The key to this lemma is that maximization of a convex function always occurs at the boundary. The proof requires more work since the maximization is executed over a non-compact space. The full proof is in the appendix.

We now show that the optimal strategy for Nature converges to the geometric Brownian motion.

\begin{theorem}\label{Thm:Convergence}
For any sequence $(S_n^\ast, n \in \N)$ with $S_n^\ast \in \S^n_{D=,\zb,\mg}$,
\begin{equation*}
X(S_n^*) \convto{d} G
\end{equation*}
where $G$ is the geometric Brownian motion defined in~\eqref{eq:gbmdef}.
\end{theorem}
\begin{proof}
For each $n \in \N$ and $0 \leq m \leq n$ define
\begin{equation*}
W_{n,m} = \log S_{n,m}^\ast + \frac{cm}{2n},
\end{equation*}
and let $W_{n,(n \cdot)} \in \C[0,1]$ denote the linear interpolation of the values $(W_{n,m}, 0 \leq m \leq n)$. That is, for $t \in [0,1]$ of the form $t = (m + \alpha)/n$ with $m \in \{0,1,\dots,n-1\}$ and $0 \leq \alpha \leq 1$, we define
\begin{equation*}
W_{n,(nt)} = (1-\alpha) \: W_{n,m} + \alpha \: W_{n,m+1} = \log X(S_n^\ast)(t) + \frac{ct}{2}.
\end{equation*}
It suffices to show that $W_{n,(n \cdot)} \convto{d} \sqrt{c}B$ where $B$ is the standard Brownian motion, for then we have
\begin{equation*}
\log X(S_n^\ast) = \left(W_{n,(n \cdot)} - \frac{ct}{2} \colon t \in [0,1] \right) \convto{d} \left( \sqrt{c} B(t) - \frac{ct}{2} \colon t \in [0,1] \right),
\end{equation*}
and thus $X(S_n^\ast) \convto{d} G$ by the continuous mapping theorem.

Throughout, let $\F_{n,m} = \sigma(S_{n,0}^\ast, \dots, S_{n,m}^\ast)$ and $T_{n,m} = \frac{S_{n,m}^\ast}{S_{n,m-1}^\ast}-1$. Let $\beta_{n,0} = 0$ and $\beta_{n,m} \in \F_{n,m-1}$ be a predictable sequence such that $M_{n,m} = W_{n,m} + \beta_{n,m}$ is a martingale sequence with respect to $(\F_{n,m})$, and let $\beta_{n,(n \cdot)}, M_{n,(n \cdot)} \in \mathcal{C}[0,1]$ be the linear interpolations of $(\beta_{n,m})$ and $(M_{n,m})$, respectively. Our approach is to show that
\begin{equation}\label{Eq:Requirement}
M_{n,(n \cdot)} \convto{d} \sqrt{c} B \quad \text{ and } \quad \beta_{n,(n \cdot)} \convto{p} \mathbf{0}
\end{equation}
where $\mathbf{0}$ is the zero function in $\mathcal{C}[0,1]$. By Slutsky's theorem~\cite[Theorems~18.10 and~18.11]{van2000asymptotic}, these would imply the desired result $W_{n,(n \cdot)} = M_{n,(n \cdot)}-\beta_{n,(n \cdot)} \convto{d} \sqrt{c} B$.

To show $M_{n,(n \cdot)} \convto{d} \sqrt{c} B$ we appeal to the Lindeberg--Feller theorem for martingales~\cite[Theorem~7.3]{durrett2004probability}, and to show $\beta_{n,(n \cdot)} \convto{p} \mathbf{0}$ we use Taylor approximation to bound the value of $\max_{1 \leq m \leq n} |\beta_{n,m}|$. The detailed proof is in Appendix~\ref{Sec:ProofConvergence}.
\end{proof}

%%%
Furthermore, we also have the convergence of the payoff values.

\begin{lemma} \label{lem:upper-bound}
For any sequence $(S_n^\ast, n \in \N)$ with $S_n^\ast \in \S^n_{D=,\zb,\mg}$,
  \[
    \lim_{n \to \infty} L_n(S_n^\ast) = \beta.
  \]
\end{lemma}
\begin{proof}
From Theorem~\ref{Thm:Convergence} we know that $X(S_n^\ast) \convto{d} G$, so in particular, by looking at the value at $t = 1$ we also have $S_{n,n}^\ast \convto{d} G(1)$. Given $M > 0$, let $g_M(x) = \min(g(x), M)$. Note that $g_M$ is a bounded continuous function, so the convergence $S_{n,n}^\ast \convto{d} G(1)$ gives us
\begin{equation*}
\lim_{n \to \infty} \E[g_M(S_{n,n}^\ast)] = \E[g_M(G(1))] \quad \text{ for each } M > 0.
\end{equation*}
Moreover, since $g_M \uparrow g$ pointwise by the monotone convergence theorem we also know that
\begin{equation*}
\lim_{M \to \infty} \E[g_M(G(1))] = \E[g(G(1))].
\end{equation*}
We claim that
\begin{equation*}
\lim_{M \to \infty} \E[g_M(S_{n,n}^\ast)] = \E[g(S_{n,n}^\ast)] \quad \text{ uniformly in } n;
\end{equation*}
we prove this claim in Appendix~\ref{sec:upper-bound}. This uniform convergence allows us to interchange the order of the limit operations, giving us the desired result
\begin{align*}
\lim_{n \to \infty} \E[g(S_{n,n}^\ast)]
&= \lim_{n \to \infty} \lim_{M \to \infty} \E[g_M(S_{n,n}^\ast)]
= \lim_{M \to \infty} \lim_{n \to \infty} \E[g_M(S_{n,n}^\ast)]
= \lim_{M \to \infty} \E[g_M(G(1))]
= \E[g(G(1))].
\end{align*}

\end{proof}

%%%%
The last lemma states that we have a matching lower bound for the limit of the game value.

\begin{lemma}\label{lem:lower-bound}
\begin{equation*}
  \beta \; \leq \; \lim_{n \to \infty} \inf_{A\in\A_n} \sup_{X \in \X^n_{C,\zb}} L_n(A,X).
\end{equation*}
\end{lemma}
\begin{proof}
For each $n \in \N$, define a stochastic process $G_n \colon [0,1] \to \R_+$ by setting $G_n = G$ if the geometric Brownian motion $G$ does not violate \ZC, and $G_n = \mathbf{1}$ (the constant function $1$) otherwise. Observe that $G_n$ is continuous and satisfies \ZC\ by construction. Furthermore, $G_n$ satisfies \CVC\ because $G_n$ is a mixture of two components, namely $G$ and the constant function $\mathbf{1}$, each of which satisfies \CVC. This shows that $G_n \in \X^n_{C,\zb}$. Similarly, $G_n$ is a martingale, so $S_n(G_n)$ is also a martingale. Thus by~\eqref{eq:player-irrelevant} we have
\begin{equation*}
\inf_{A\in\A_n} \sup_{X \in \X^n_{C,\zb}} L_n(A,X)
\geq \inf_{A\in\A_n} L_n(A,G_n)
= L_n(G_n)
= \E[g(G_n(1))].
\end{equation*}
We now show that $\lim_{n \to \infty} \E[g(G_n(1))] = \E[g(G(1))]$, which will give us the desired result:
\begin{equation*}
\liminf_{n \to \infty} \inf_{A\in\A_n} \sup_{X \in \X^n_{C,\zb}} L_n(A,X)
\geq \lim_{n \to \infty} \E[g(G_n(1))]
= \E[g(G(1))] = \beta.
\end{equation*}

Recall that $\E[G(1)] = 1$, so $\E[(G(1)-1)^2] = \Var(G(1)) = \exp(c)-1$. We claim that for sufficiently large $n$,
\begin{equation}\label{eq:probconv}
\P(G \text{ does not violate \ZC})  \geq \left(1-\frac{1}{n^2}\right)^n \to 1 \quad \text{ as } n \to \infty.
\end{equation}
We prove this claim in Appendix~\ref{app:lower-bound}. Then using Lipschitz property of $g$ and the Cauchy-Schwarz inequality, we obtain
\begin{equation*}
\begin{split}
\big|\E\big[g(G(1))\big]-\E\big[g(G_n(1))\big]\big|
&= \big|\E\big[(g(G(1))-g(1)) \:\cdot \mathbf{1}\{G \text{ violates \ZC}\} \big] \big| \\
&\leq \E\big[|g(G(1))-g(1)| \cdot \mathbf{1}\{G \text{ violates \ZC}\} \big] \\
&\leq L \: \E\big[|G(1)-1| \cdot \mathbf{1}\{G \text{ violates \ZC}\} \big] \\
&\leq L \: \E[(G(1)-1)^2]^{1/2} \: \P(G \text{ violates \ZC})^{1/2} \\
&\leq L \: \big(\exp(c)-1\big)^{1/2} \: \left(1-\left(1-\frac{1}{n^2}\right)^n\right)^{1/2} \\
&\to 0 \quad \text{ as } n \to \infty.
\end{split}
\end{equation*}
\end{proof}

%%%%%%%%%%
\bibliographystyle{plain} 
\bibliography{option_pricing}

%%%%%%%%%%
\appendix

%-----------
\section{Proof of Lemma~\ref{lem:minimax}}

Once we write out the objective $L_n(A,X)$ explicitly, we see that the single $\inf \sup$ can be broken down into a sequence of nested $\sup \inf$'s as follows. Let $T_{n,m}$ be the random variable $\frac{X(m/n)}{X((m-1)/n)} - 1$ for $m=1, \ldots n$. Then we have
  \begin{multline} \label{eq:infsups}
  \inf_{A\in\A_n} \sup_{S_n \in \S^n_{D,\zb}} L_n(A,S_n)
  \quad = \quad
  \inf_{\Delta_1} \sup_{T_{n,1}} \E_1 \left[ -T_{n,1} \Delta_1 \;+\;
    \inf_{\Delta_2} \sup_{T_{n,2}} \E_2 \left[
      - T_{n,2}\Delta_2 \eqnspacer \right. \right.  \\
    + \cdots +
    \left.\left. \inf_{\Delta_n} \sup_{T_{n,n}} \E_n \left[
      -T_{n,n} \Delta_n + g\left(\prod_{m=1}^n (1 + T_{n,m})\right) \eqnspacer
      \right] \cdots \right] \right],
\end{multline}
where $\E_m$ should be interpreted to mean the expectation conditioned on the filtration of the random path $X$ up to time $\frac{m-1}{n}$.

We now recall a simplified version of Sion's minimax theorem \cite{sion1958general}, a generalization of Von Neumann's minimax theorem \cite{von_neumann_theory_1947}. Assume we are given a compact convex set $\Lambda$ and a convex set $\Omega$, each a subset of a linear topological space, and we have a function $f : \Omega \times \Lambda \to \reals$ continuous in each input, convex in the first input and concave in the second. Then it holds that 
\[
  \inf_{\omega \in \Omega} \sup_{\lambda \in \Lambda} f(\omega,\lambda) = 
  \sup_{\lambda \in \Lambda} \inf_{\omega \in \Omega} f(\omega,\lambda).
\]
This can be applied to each of the nested $\inf\sup$'s in \eqref{eq:infsups} recursively, where we substitute $\Omega = \reals$ and $\Lambda = \Delta([-\zeta_n, \zeta_n])$, the set of distributions on the interval which is compact in the weak topology\footnote{Notice that the use of \ZC\ is critical here, as we wouldn't have compactness otherwise.}. The objective function of the $m$-th nested $\inf\sup$ is clearly linear in $\Delta_m$, and it is also linear in the distribution on $T_{n,m}$ as we are simply taking an expectation over this distribution. We then have
\begin{multline} \label{eq:infsups}
  \inf_{A\in\A_n} \sup_{S_n \in \S^n_{D,\zb}} L_n(A,S_n)
  \quad = \quad
  \sup_{T_{n,1}} \inf_{\Delta_1} \E_1 \left[ -T_{n,1} \Delta_1 \;+\;
    \sup_{T_{n,2}} \inf_{\Delta_2} \E_2 \left[
      - T_{n,2}\Delta_2 \eqnspacer \right. \right.  \\
    + \cdots +
    \left.\left. \sup_{T_{n,n}} \inf_{\Delta_n} \E_n \left[
      -T_{n,n} \Delta_n + g\left(\prod_{m=1}^n (1 + T_{n,m})\right) \eqnspacer
      \right] \cdots \right] \right],
\end{multline}
Of course, since we are considering $\Delta_m$ as a function of the history up to time $(m-1)$, Nature may as well solve for the optimal randomized price path determined by $T_{n,1}, \ldots, T_{n,n}$ in advance and announce this to the investor. Since $\Delta_m$ does not interact with any price fluctuations beyond the $m$-th one, the choice of $\Delta_m$ may as well be made with knowledge of $T_{n,m+1}, \ldots, T_{n,n}$. In other words, the sequence of $\sup$'s can be gathered together, and we obtain the desired result.

%--------
\section{Proof of Lemma~\ref{lem:maxvar}}

We will assume without loss of generality that the filtration $\F_{n,m}$ is fixed for the set $\S^n_{D,\zb,\mg}$.  Taking the sup over the weak topology, we see that the argsup is nonempty since the set $\S^n_{D,\zb,\mg}$ is compact.  Define
\[
  M(S_n) := \{m \colon \E[T_{n,m}^2 \mid \F_{n,m-1}]<\exp(c/n)-1\},
\]
where $T_{n,m}$ is defined as usual.  In other words, $M(S_n)$ is set of steps $m$ where \DVC\ is slack for $S_n$.  Now let $\hat S_n$ be an element of the argsup such that $|M(\hat S_n)|$ is minimal.  Assume for a contradiction that $M(\hat S_n)$ is nonempty, and let $m^* \in M(\hat S_n)$.

We wish to construct some $S^*_n$ from $\hat S_n$ by modifying $\hat T_{n,m^*}$.  Note that since the filtration $\F_{n,m}$ is fixed, for each $m$ the constraints \DVC, \ZC, and \mg\ on $\hat T_{n,m}$ are independent of the values of $\hat T_{n,m^*}$, so to ensure $S^*_n \in \S^n_{D,\zb,\mg}$, we need only maintain the three constraints for the modified $\hat T_{n,m^*}$.

For brevity let $\hat T = \hat T_{n,m^*}$.  Denote $v^* = \exp(c/n) - 1$ and $v = \E[\hat T^2 \mid \F_{n,m^*-1}]$.  Note that $v < v^*$ by assumption, and for sufficiently large $n$, $v^* < \zeta_n^2$ by the definition of $\zb$.  Henceforth we will assume $n$ is large enough for the latter inequality.  Let $A$ be an independent event with $\P(A) = \alpha$, where $\alpha := (v^*-v)/(\zeta_n^2-v)$, and let $Z$ be the random variable which is $\zeta_n$ and $-\zeta_n$ each with probability $1/2$.  Finally, set $T = \hat T \mathbf{1}_{\bar A} + Z \mathbf{1}_A$ and define $S^*_n$ by $T^*_{n,m} = \hat T_{n,m}$ for $m\neq m^*$ and $T^*_{n,m^*} = T$.  Note that this $T$ satisfies \ZC\ and \mg\ trivially, and satisfies \DVC\ with equality:
\[
  \E[T^2 \mid \F_{n,m^\ast-1}] = (1-\alpha) \E[\hat T^2 \mid \F_{n,m^\ast-1}] + \alpha \zeta_n^2 = v + \alpha (\zeta_n^2 - v) = v^*,
\]
and hence $S^*_n \in \S^n_{D,\zb,\mg}$ and in fact $|M(S^*_n)| = |M(\hat S_n)| - 1$.  We will show below that $\E[g(\hat S_{n,n})] \leq \E[g(S^*_{n,n})]$, meaning $S^*_n$ is also in the argsup, thus condtradicting the minimality of $|M(\hat S_n)|$.  Hence, for sufficiently large $n$, we can select some $S^*_n$ in the argsup such that $M(S^*_n) = \emptyset$, which completes the proof.

We now show $\E[g(\hat S_{n,n})] \leq \E[g(S^*_{n,n})]$.  Observe that $g\left(\prod_{m=1}^n (1+t_m)\right)$ is convex in each $t_m$ (fixing the others).  Thus, we see that conditioned on $\{\hat T_{n,m} \mid m \neq m^\ast\}$,
\[
  f(t) = g\left((1+t)\prod_{m\neq m^*} (1+\hat T_{n,m})\right) - at
\]
is convex in $t$, where $a$ is chosen such that $f(\zeta_n) = f(-\zeta_n)$.  Now we have
\[
  \E[f(T)] - \E[f(\hat T)] = \alpha\left(\E[f(Z)] - \E[f(\hat T)]\right) = \alpha \left( f(\zeta_n) - \E[f(\hat T)] \right) \geq 0,
\]
since $f(t) \leq f(\zeta_n)$ for all $t \in [-\zeta_n, \zeta_n]$ by convexity of $f$.  Hence, denoting $\{\hat T_{n,m} \mid m \neq m^*\}$ by $\hat T_{n,-m^*}$, we have
\[
  \E[g(\hat S_{n,n})] = \E\big[ \E[f(\hat T) \mid \hat T_{n,-m^*}] \big] \leq \E\big[\E[f(T) \mid \hat T_{n,-m^*}]\big] = \E[g(S^*_{n,n})],
\]
where we use the fact that the linear term $aT$ in the definition of $f$ does not change the expectation, as $\E[T] = \E[\hat T] = 0$.

%---------
\section{Proof of Theorem~\ref{Thm:Convergence}}\label{Sec:ProofConvergence}

For simplicity we will write $S_n$ in place of $S_n^\ast$. Note that since $S_n \in \S^n_{D=,\zb,\mg}$ we know that for each $n \in \N$ and $1 \leq m \leq n$,
\begin{equation}\label{Eq:SecondMoment}
\E[T_{n,m} \mid \F_{n,m-1}] = 0, \quad \E[T_{n,m}^2 \mid \F_{n,m-1}] = \exp(c/n)-1, \quad \text{ and } \quad |T_{n,m}| \leq \zeta_n \quad \text{ a.s.}
\end{equation}

\paragraph{Showing $\beta_{n,(n \cdot)} \convto{p} \mathbf{0}$.}
Observe that for each $n \in \N$ and $1 \leq m \leq n$,
\begin{equation*}
\begin{split}
\beta_{n,m} &= -\E[W_{n,m}-W_{n,m-1} \mid \F_{n,m-1}] + \beta_{n,m-1} \\
&= -\E\left[ \log\left(\frac{S_{n,m}}{S_{n,m-1}}\right) \;\Big|\; \F_{n,m-1} \right] - \frac{c}{2n} + \beta_{n,m-1} \\
&= -\E\left[ \log\left(1 + T_{n,m}\right) \mid \F_{n,m-1} \right] - \frac{c}{2n} + \beta_{n,m-1}.
\end{split}
\end{equation*}
Now write
\begin{equation*}
\beta_{n,m} = \sum_{k=1}^m (\beta_{n,k}-\beta_{n,k-1})
= -\sum_{k=1}^m \left(\E[\log(1+T_{n,k}) \mid \F_{n,k-1}] + \frac{c}{2n} \right),
\end{equation*}
so by Lemma~\ref{Lem:LogS1}, for sufficiently large $n$ and for all $1 \leq m \leq n$ we have
\begin{equation*}
|\beta_{n,m}|
\leq \sum_{k=1}^m \left|\E[\log(1+T_{n,k}) \mid \F_{n,k-1}] + \frac{c}{2n} \right|
\leq \sum_{k=1}^m \left(\frac{2c \: \zeta_n}{n} + \frac{c^2}{n^2}\right)
\leq 2c \: \zeta_n + \frac{c^2}{n}.
\end{equation*}
Thus, since $\beta_{n,(n \cdot)}$ is a linear interpolation of $(\beta_{n,m}, 0 \leq m \leq n)$, for sufficiently large $n$ we have
\begin{equation*}
\max_{0 \leq t \leq 1} |\beta_{n,(nt)}| = \max_{1 \leq m \leq n} |\beta_{n,m}| \leq 2c \: \zeta_n + \frac{c^2}{n} \to 0 \quad \text{ as } n \to \infty.
\end{equation*}
This shows that in fact $\beta_{n,(n \cdot)} \to \mathbf{0}$ a.s.

\paragraph{Showing $M_{n,(n \cdot)} \convto{d} \sqrt{c}B$.}
For $1 \leq m \leq n$ let $Y_{n,m}$ be the martingale differences,
\begin{equation*}
\begin{split}
Y_{n,m} &= M_{n,m}-M_{n,m-1} \\
&= W_{n,m}-W_{n,m-1} -\E[W_{n,m}-W_{n,m-1} \mid \F_{n,m-1}] \\
&= \log(1+T_{n,m}) - \E[\log(1+T_{n,m}) \mid \F_{n,m-1}].
\end{split}
\end{equation*}
Let $V_{n,0} = 0$ and for $1 \leq m \leq n$, let $V_{n,m}$ be the partial sum of the conditional variance,
\begin{equation*}
V_{n,m} = \sum_{k=1}^m \E[Y_{n,k}^2 \mid \F_{n,k-1}] = \sum_{k=1}^m \Var(\log(1+T_{n,k}) \mid \F_{n,k-1}).
\end{equation*}
By the Lindeberg-Feller theorem for martingales~\cite[Theorem~7.3]{durrett2004probability}, to prove $M_{n,(n \cdot)} \convto{d} \sqrt{c}B$ it suffices to show that
\begin{enumerate}
  \item for all $\varepsilon > 0$, $\sum_{m=1}^n \E[Y_{n,m}^2 \: \mathbf{1}\{|Y_{n,m}| > \varepsilon\} \mid \F_{n,m-1}] \convto{p} 0$, and
  \item $V_{n, \lfloor nt \rfloor} \convto{p} ct$ for all $0 \leq t \leq 1$.
\end{enumerate}
The first condition is easy to satisfy using Lemma~\ref{Lem:LogS1}. Indeed, given $\varepsilon > 0$, from Lemma~\ref{Lem:LogS1} we see that for sufficiently large $n$ and for all $1 \leq m \leq n$,
\begin{equation*}
\big|\E[\log(1+T_{n,m}) \mid \F_{n,m-1}] \big| \leq \frac{2c \: \zeta_n}{n} + \frac{c^2}{n^2} \leq \frac{\varepsilon}{2} \quad \text{ a.s.}
\end{equation*}
Moreover, from the assumption that $|T_{n,m}| \leq \zeta_n \to 0$, for sufficiently large $n$ and for all $1 \leq m \leq n$ we also have
\begin{equation*}
|\log(1+T_{n,m})| \leq \zeta_n \leq \frac{\varepsilon}{2} \quad \text{ a.s.}
\end{equation*}
Thus for sufficiently large $n$ and for all $1 \leq m \leq n$,
\begin{equation*}
|Y_{n,m}| \leq |\log(1+T_{n,m})| + \big|\E[\log(1+T_{n,m}) \mid \F_{n,m-1}] \big| \leq \varepsilon \quad \text{ a.s.},
\end{equation*}
which implies the asymptotic negligibility condition,
\begin{equation*}
\sum_{m=1}^n \E[Y_{n,m}^2 \: \mathbf{1}\{|Y_{n,m}| > \varepsilon\} \mid \F_{n,m-1}] = 0 \quad \text{ a.s.\ for sufficiently large } n.
\end{equation*}
For the second condition, let $0 \leq t \leq 1$ be given. Then by Lemma~\ref{Lem:X}, for sufficiently large $n$ we have
\begin{equation*}
\begin{split}
\left| V_{n, \lfloor nt \rfloor} - ct \right|
&= \left| \sum_{m=1}^{\lfloor nt \rfloor} \Var(\log(1+T_{n,m}) \mid \F_{n,m-1}) - ct \right| \\
&\leq \sum_{m=1}^{\lfloor nt \rfloor} \left| \Var(\log(1+T_{n,m}) \mid \F_{n,m-1}) - \frac{c}{n} \right| + \left| \frac{c \lfloor nt \rfloor}{n} - ct \right| \\
&\leq \sum_{m=1}^{\lfloor nt \rfloor} \left(\frac{4c \: \zeta_n}{n} + \frac{3c^2}{n^2}\right) + \left(\frac{c}{n} (nt+1) - ct \right) \\
&\leq 4c \: \zeta_n + \frac{3c^2}{n} + \frac{c}{n} \\
&\to 0 \quad \text{ as } n \to \infty.
\end{split}
\end{equation*}
Thus $V_{n, \lfloor nt \rfloor} \to ct$ a.s.\ for each $0 \leq t \leq 1$.

To complete the proof of Theorem~\ref{Thm:Convergence} we establish the following lemmas.

%-----
\begin{lemma}\label{Lem:LogS1}
For sufficiently large $n$ and for all $1 \leq m \leq n$ we have
\begin{equation*}
\left|\E\left[\log(1+T_{n,m}) \mid \F_{n,m-1} \right] + \frac{c}{2n} \right| \leq \frac{2c \: \zeta_n}{n} + \frac{c^2}{n^2} \quad \text{ a.s.}
\end{equation*}
\end{lemma}

\begin{proof}
Observe that for sufficiently large $n$, say $n \geq N$, and for all $1 \leq m \leq n$, the following inequalities hold:
\begin{equation*}
\exp(c/n)-1 \leq \frac{2c}{n},
\quad \exp(c/n)-1-\frac{c}{n} \leq \frac{2c^2}{n^2},
\quad\text{ and }\quad 3(1-\zeta_n) \geq 1.
\end{equation*}
Throughout the rest of this proof, all inequalities hold for $n \geq N$ and uniformly for all $1 \leq m \leq n$. 

By the second-order Taylor expansion (with remainder) of the function $x \mapsto \log(1+x)$ around the point $x = 0$,
\begin{equation*}
\log(1+T_{n,m})-T_{n,m}+\frac{1}{2} T_{n,m}^2 = \frac{1}{3(1+\xi_{n,m})^3} T_{n,m}^3
\end{equation*}
where $\xi_{n,m}$ is some value between $0$ and $T_{n,m}$. Then, since $|T_{n,m}| \leq \zeta_n$,
\begin{equation*}
\begin{split}
\left|\log(1+T_{n,m}) - T_{n,m} + \frac{1}{2} T_{n,m}^2 \right|
&= \left|\frac{1}{3(1+\xi_{n,m})^3} T_{n,m}^3 \right|
\leq \frac{|T_{n,m}|^3}{3(1-\zeta_n)^3}
\leq \zeta_n \: T_{n,m}^2,
\end{split}
\end{equation*}
so
\begin{equation*}
\begin{split}
\E\Bigg[\Big|\log(1+T_{n,m}) - T_{n,m} + \frac{1}{2} T_{n,m}^2 \Big|  \;\Bigg|\; \F_{n,m-1} \Bigg]
\leq \zeta_n \: \E[T_{n,m}^2 \mid \F_{n,m-1}]
= \zeta_n \big( \exp(c/n)-1 \big)
\leq \frac{2c \: \zeta_n}{n}.
\end{split}
\end{equation*}
Therefore,
\begin{equation*}
\begin{split}
\left| \E[\log(1+T_{n,m}) \mid \F_{n,m-1}] + \frac{c}{2n} \right|
&= \Bigg| \E\Big[ \log(1+T_{n,m}) - T_{n,m} + \frac{1}{2} T_{n,m}^2 \;\Big|\; \F_{n,m-1} \Big] -\frac{1}{2}\Big(\exp(c/n)-1-\frac{c}{n}\Big) \Bigg| \\
&\leq \E\Bigg[ \Big|\log(1+T_{n,m}) - T_{n,m} + \frac{1}{2} T_{n,m}^2 \Big| \;\Bigg|\; \F_{n,m-1} \Bigg] + \frac{1}{2} \Bigg| \exp(c/n)-1-\frac{c}{n} \Bigg| \\
&\leq \frac{2c \: \zeta_n}{n} + \frac{c^2}{n^2}.
\end{split}
\end{equation*}
\end{proof}

%------
\begin{lemma}\label{Lem:LogS2}
For sufficiently large $n$ and for all $1 \leq m \leq n$ we have
\begin{equation*}
\left|\E\left[\log^2(1+T_{n,m}) \mid \F_{n,m-1} \right] - \frac{c}{n} \right| \leq \frac{4c \: \zeta_n}{n} + \frac{2c^2}{n^2} \quad \text{ a.s.}
\end{equation*}
\end{lemma}

\begin{proof}
We follow a similar argument as in the proof of Lemma~\ref{Lem:LogS1}. Observe that since $\zeta_n \to 0$, for sufficiently large $n$, say $n \geq N$, the following inequalities hold:
\begin{equation*}
\exp(c/n)-1 \leq \frac{2c}{n},
\quad \exp(c/n)-1-\frac{c}{n} \leq \frac{2c^2}{n^2},
\quad\text{ and }\quad \frac{3 - 2\log(1-\zeta_n)}{3(1-\zeta_n)^3} \leq 2.
\end{equation*}
Throughout the rest of this proof, all inequalities hold for $n \geq N$ and uniformly for all $1 \leq m \leq n$.

Recall that from the second-order Taylor expansion of the function $x \mapsto \log^2(1+x)$ around the point $x = 0$,
\begin{equation*}
\log^2(1+T_{n,m}) - T_{n,m}^2 = \left(\frac{-3+2\log(1+\xi_{n,m})}{3(1+\xi_{n,m})^3}\right) \: T_{n,m}^3
\end{equation*}
where $\xi_{n,m}$ is some value between $0$ and $T_{n,m}$. Then, since $|T_{n,m}| \leq \zeta_n$,
\begin{equation*}
\begin{split}
\left|\log^2(1+T_{n,m}) - T_{n,m}^2 \right|
&= \left|\left(\frac{-3+2\log(1+\xi_{n,m})}{3(1+\xi_{n,m})^3}\right) \: T_{n,m}^3 \right| \\
&\leq \frac{3 - 2\log(1-|T_{n,m}|)}{3(1-|T_{n,m}|)^3} \: |T_{n,m}|^3 \\
&\leq \frac{3 - 2\log(1-\zeta_n)}{3(1-\zeta_n)^3} \: \zeta_n \: T_{n,m}^2 \\
&\leq 2 \zeta_n \: T_{n,m}^2.
\end{split}
\end{equation*}
Then
\begin{equation*}
\begin{split}
\E\Big[ \big| \log^2(1+T_{n,m}) - T_{n,m}^2 \big| \;\Big|\; \F_{n,m-1} \Big]
&\leq 2\zeta_n \: \E[T_{n,m}^2 \mid \F_{n,m-1}]
= 2\zeta_n \big( \exp(c/n)-1 \big) 
\leq \frac{4c \: \zeta_n}{n}.
\end{split}
\end{equation*}
Therefore,
\begin{equation*}
\begin{split}
\left|\E\left[\log^2(1+T_{n,m}) \mid \F_{n,m-1} \right] - \frac{c}{n} \right|
&= \left|\E\left[\log^2(1+T_{n,m})-T_{n,m}^2 \mid \F_{n,m-1} \right] + \Big(\exp(c/n)-1-\frac{c}{n} \Big) \right| \\
&\leq \E\Big[ \big| \log^2(1+T_{n,m}) - T_{n,m}^2 \big| \;\Big|\; \F_{n,m-1} \Big] + \left|\exp(c/n)-1-\frac{c}{n}\right| \\
&\leq \frac{4c \: \zeta_n}{n} + \frac{2c^2}{n^2}.
\end{split}
\end{equation*}
\end{proof}

%------
\begin{lemma}\label{Lem:X}
For sufficiently large $n$ and for all $1 \leq m \leq n$ we have
\begin{equation*}
\left|\Var(\log(1+T_{n,m}) \mid \F_{n,m-1})-\frac{c}{n}\right| \leq \frac{4c \: \zeta_n}{n} + \frac{3c^2}{n^2} \quad \text{ a.s.}
\end{equation*}
\end{lemma}
\begin{proof}
Writing
\begin{equation*}
\Var(\log(1+T_{n,m}) \mid \F_{n,m-1}) = \E[\log^2(1+T_{n,m}) \mid \F_{n,m-1}] - \E[\log(1+T_{n,m}) \mid \F_{n,m-1}]^2,
\end{equation*}
we can bound
\begin{equation*}
\begin{split}
\left|\Var(\log(1+T_{n,m}) \mid \F_{n,m-1})-\frac{c}{n}\right|
&\leq \left|\E[\log^2(1+T_{n,m}) \mid \F_{n,m-1}]-\frac{c}{n}\right| + \E[\log(1+T_{n,m}) \mid \F_{n,m-1}]^2.
\end{split}
\end{equation*}
For sufficiently large $n$ and for all $1 \leq m \leq n$, by Lemma~\ref{Lem:LogS2} the first term above is at most $4c \: \zeta_n/n + 2c^2/n^2$, while by Lemma~\ref{Lem:LogS1} the first term above is at most
\begin{equation*}
\left(\frac{c}{2n} + \frac{2c \: \zeta_n}{n} + \frac{c^2}{n^2} \right)^2 \leq \frac{c^2}{n^2}.
\end{equation*}
Thus for sufficiently large $n$ and for all $1 \leq m \leq n$,
\begin{equation*}
\left|\Var(\log(1+T_{n,m}) \mid \F_{n,m-1})-\frac{c}{n}\right|
\leq \frac{4c \: \zeta_n}{n} + \frac{2c^2}{n^2} + \frac{c^2}{n^2}
= \frac{4c \: \zeta_n}{n} + \frac{3c^2}{n^2},
\end{equation*}
as desired.
\end{proof}

%---------
\section{Proof of Lemma~\ref{lem:upper-bound}}
\label{sec:upper-bound}

In view of the proof sketch in Section~\ref{sec:main-result}, it suffices to show that
\begin{equation*}
\lim_{M \to \infty} \E[g_M(S_{n,n}^\ast)] = \E[g(S_{n,n}^\ast)] \quad \text{ uniformly in } n.
\end{equation*}
For simplicity we will write $S_n$ in place of $S_n^\ast$. Without loss of generality we may assume $g(0) = 0$, so the Lipschitz property of $g$ tells us that $g(x) \leq Lx$ for all $x \geq 0$.

Recall that since $S_n$ is a martingale we have $\E[S_{n,n}] = \E[S_{n,0}] = 1$ for all $n \in \N$. We now show that for each $n \in \N$,
\begin{equation}\label{eq:var-Snm}
\E[S_{n,m}^2] = \exp\left(\frac{cm}{n}\right) \quad \text{ for all } m = 1,\dots,n.
\end{equation}
We proceed by induction on $m$, for each fixed $n$. The base case $m = 1$ follows from the fact that $S_n \in \S^n_{D=,\zb,\mg}$, so $S_n$ satisfies \DVC\ with equality. Now assume the claim is true for some $1 \leq m < n$. For $m+1$,  by expanding the \DVC\ constraint
\begin{equation*}
\E[(S_{n,m+1}-S_{n,m})^2 \mid \F_{n,m}] = \left(\exp(c/n)-1\right) S_{n,m}^2
\end{equation*}
we get
\begin{equation*}
\E[S_{n,m+1}^2 \mid \F_{n,m}] = 2S_{n,m} \: \E[S_{n,m+1} \mid \F_{n,m}] - S_{n,m}^2 + \left(\exp(c/n)-1\right) S_{n,m}^2 = \exp(c/n) \: S_{n,m}^2.
\end{equation*}
Taking expectation on both sides and using the inductive hypothesis, we obtain
\begin{equation*}
\E[S_{n,m+1}^2] = \exp\left(\frac{c}{n} \right) \: \E[S_{n,m}^2] = \exp\left(\frac{c(m+1)}{n} \right),
\end{equation*}
which proves~\eqref{eq:var-Snm}. In particular, by plugging in $m = n$ to~\eqref{eq:var-Snm} we see that $\E[S_{n,n}^2] = \exp(c)$ for all $n \in \N$.

Now fix $1 \leq n \leq \infty$. By the Cauchy-Schwarz inequality, for $M > 0$ we have
\begin{equation*}
\begin{split}
\big| \E[g(S_{n,n})] - \E[g_M(S_{n,n})] \big|
&= \E[(g(S_{n,n})-M) \cdot \mathbf{1}\{g(S_{n,n}) > M\}] \\
&\leq \E[g(S_{n,n}) \cdot \mathbf{1}\{g(S_{n,n}) > M\}] \\
&\leq \E[g(S_{n,n})^2]^{1/2} \; \P(g(S_{n,n}) > M)^{1/2}
\end{split}
\end{equation*}
For the first factor, since $g(S_{n,n}) \leq LS_{n,n}$ we have that $\E[g(S_{n,n})^2] \leq L^2 \: \E[S_{n,n}^2] = L^2 \exp(c)$. Similarly, for the second factor, by Markov inequality we have
\begin{equation*}
\P(g(S_{n,n}) > M) \leq \P(S_{n,n} > M/L) \leq \frac{\E[S_{n,n}]}{M/L} = \frac{L}{M}.
\end{equation*}
Therefore,
\begin{equation*}
\big| \E[g(S_{n,n})] - \E[g_M(S_{n,n})] \big| \leq \frac{L^{3/2} \: \exp(c/2)}{M^{1/2}},
\end{equation*}
and since the bound is independent of $n$, this shows that $\lim_{M \to \infty} \E[g_M(S_{n,n})] = \E[g(S_{n,n})]$ uniformly in $n$, as desired.

%------------
\section{Proof of Lemma~\ref{lem:lower-bound}}
\label{app:lower-bound}

In view of the proof sketch in Section~\ref{sec:main-result}, it remains to show that for sufficiently large $n$,
\begin{equation*}
\P(G \text{ does not violate \ZC})  \geq \left(1-\frac{1}{n^2}\right)^n.
\end{equation*}

Recall that $G(t) = \exp(\sqrt{c}B(t) - ct/2)$ where $B$ is the standard Brownian motion. Thus for each $n \in \N$ and for all $0 \leq m \leq n-1$,
\begin{equation*}
\frac{G((m+1)/n)}{G(m/n)} = \exp\left(\sqrt{c}B\left(\frac{m+1}{n}\right) - \sqrt{c} B\left(\frac{m}{n}\right) - \frac{c}{2n} \right) = \exp\left(\sqrt{\frac{c}{n}} Z_m - \frac{c}{2n} \right)
\end{equation*}
where $Z_m = \sqrt{n}(B(\frac{m+1}{n}) - B(\frac{m}{n}))$ has $N(0,1)$ distribution and $Z_0,Z_1,\dots,Z_{n-1}$ are independent. Therefore,
\begin{equation*}
\begin{split}
\P(G \text{ does not violate \ZC})
&= \P\left(\max_{0 \leq m \leq n-1} \left|\frac{G((m+1)/n)}{G(m/n)}-1\right| \leq \zeta_n \right) \\
&= \P\left(\max_{0 \leq m \leq n-1} \left|\exp\left(\sqrt{\frac{c}{n}} Z_m \right)-1\right| \leq \zeta_n \right) \\
&= \P\left(\left|\exp\left(\sqrt{\frac{c}{n}} Z - \frac{c}{2n} \right)-1\right| \leq \zeta_n \right)^n
\end{split}
\end{equation*}
where $Z \sim N(0,1)$.

From the assumptions $\zeta_n \to 0$ and $\liminf_{n \to \infty} n \zeta_n^2/\log n > 16c$ we can choose $n$ large enough, say $n \geq N$, such that the following inequalities are true:
\begin{equation*}
\begin{split}
\log(1+\zeta_n) \geq \frac{\zeta_n}{2},
\quad\quad
\sqrt{\frac{n}{c}} \zeta_n \geq \sqrt{\frac{c}{n}},
\quad\quad
\frac{\sqrt{n} \: \zeta_n}{2 \: \sqrt{c}} \geq \sqrt{\frac{2}{\pi}}
\quad \text{ and }
\quad
\frac{n \: \zeta_n^2}{\log n} \geq 16c.
\end{split}
\end{equation*}
Throughout the remainder of this proof we suppose $n \geq N$. Observe that if we have $\exp\left(\sqrt{\frac{c}{n}} Z-\frac{c}{2n} \right)-1 > \zeta_n$ then
\begin{equation*}
Z > \frac{1}{2} \sqrt{\frac{c}{n}} + \sqrt{\frac{n}{c}} \: \log(1+\zeta_n)
> \frac{\sqrt{n} \: \zeta_n}{2\sqrt{c}},
\end{equation*}
and similarly, if $\exp\left(\sqrt{\frac{c}{n}} Z-\frac{c}{2n} \right)-1 < - \zeta_n$ then
\begin{equation*}
Z < \frac{1}{2} \sqrt{\frac{c}{n}} + \sqrt{\frac{n}{c}} \: \log(1-\zeta_n)
\leq \frac{1}{2} \sqrt{\frac{c}{n}} - \sqrt{\frac{n}{c}} \: \zeta_n
\leq -\frac{\sqrt{n} \: \zeta_n}{2\sqrt{c}}.
\end{equation*}
This shows that
\begin{equation*}
\P\left(\left|\exp\left(\sqrt{\frac{c}{n}} Z - \frac{c}{2n} \right)-1\right| > \zeta_n \right)
\leq \P\left(|Z| > \frac{\sqrt{n} \: \zeta_n}{2\sqrt{c}} \right)
= 2 \: \P\left(Z > \frac{\sqrt{n} \: \zeta_n}{2\sqrt{c}} \right).
\end{equation*}
By standard Gaussian tail bound~\cite[Theorem~1.4]{durrett2004probability},
\begin{equation*}
\P\left(Z > \frac{\sqrt{n} \: \zeta_n}{2\sqrt{c}} \right) \leq \frac{1}{\sqrt{2\pi}} \: \frac{2\sqrt{c}}{\sqrt{n} \: \zeta_n} \: \exp\left( -\frac{n\zeta_n^2}{8c} \right) \leq \frac{1}{2} \: \exp(-2\log n) = \frac{1}{2n^2}.
\end{equation*}
Thus
\begin{equation*}
\P(G \text{ does not violate \ZC}) = \left(1-2 \: \P\left(Z > \frac{\sqrt{n} \: \zeta_n}{2\sqrt{c}} \right) \right)^n \geq \left(1-\frac{1}{n^2}\right)^n.
\end{equation*}
\end{document}